\documentclass{amsart} 

\usepackage{amsmath}
\usepackage{amsfonts}
\usepackage{amssymb}
\usepackage{epsfig}
\usepackage{verbatim}
\usepackage{color}
\usepackage{graphicx}
\usepackage[letterpaper]{geometry}
\usepackage[numbers,square]{natbib}

\let\phi\varphi
\let\hat\widehat
\newcommand{\fp}{(S^*,T^*)}

\newcommand{\reg}{\text{Reg}}
\newcommand{\unreg}{\text{Unreg}}
\newcommand{\starves}{\succ}
\newcommand{\str}{\textup {Str}}
\newcommand{\unstr}{\textup {Unstr}}
\newcommand{\incomp}{\text{Incomp}}
\newcommand{\se}{\text{SE}}

\renewcommand{\bar}{\overline}

\renewcommand{\S}{S^*}

\newtheorem{theorem}{Theorem}
\newtheorem{proposition}[theorem]{Proposition}
\newtheorem{definition}[theorem]{Definition}
\newtheorem{remark}[theorem]{Remark}
\newtheorem{corollary}[theorem]{Corollary}
\newenvironment{mydef}{\begin{definition} \rm}{\end{definition}}
\newenvironment{myremark}{\begin{remark} \rm}{\end{remark}}

\definecolor{Edgar}{rgb}{1,0,0}
\definecolor{Comment}{rgb}{0.5,0,0.5}

\begin{document}

\title{Saturation effects on T-cell activation in a model of a multi-stage pathogen}

\author{Michael Shapiro}
\address{Pathology Department \\
  Tufts University \\
  Boston, MA USA }
\email{michael.shapiro@tufts.edu}
\thanks{This work was partially supported by NIH grant K25AI079404 to MS.}

\author{Edgar Delgado-Eckert}
\address{ Children's Hospital (UKBB) \\
           University of Basel \\
           Spitalstr. 33, Postfach \\
           4031 Basel \\
           Switzerland}
\email{edgar.delgado-eckert@bsse.ethz.ch}



\begin{abstract}
In \cite{cyclicPathogen}, we studied host response to a pathogen which
uses a cycle of immunologically distinct stages to establish and
maintain infection. We showed that for generic parameter values, the
system has a unique biologically meaningful stable fixed point. That
paper used a simplified model of T-cell activation, making
proliferation depend linearly on antigen-T-cell encounters. Here we
generalize the way in which T-cell proliferation depends on the sizes
of the antigenic populations.  In particular, we allow this response
to become saturated at high levels of antigen. As a result, we show
that this family of generalized models shares the same steady-state
behavior properties with the simpler model contemplated in
\cite{cyclicPathogen}.
\end{abstract}

\maketitle

\section{Introduction}

Pathogens that cyclically traverse different stages during their life cycle
or during an infection process have been studied since the late nineteenth
century. Important examples are \textit{Plasmodium} (\cite{MalariaReview}),
Trypanosoma (\cite{ChagasReview}), and the family of herpes viruses,
including the Epstein-Barr virus (EBV) (\cite{ThorleyLawson2008195}, \cite
{EBVreview}). One remarkable characteristic of infections with many of such
pathogens is life-long persistent infection (\cite{[17]}, \cite
{ThorleyLawson2008195}, \cite{[18]}, \cite{MalariaReview}, \cite
{ChagasReview}).

In \cite{cyclicPathogen}, we introduced a model of a pathogen that uses a
cycle of $n$ antigenically distinct stages to establish and maintain
infection. The model is given by $2n$ differential equations, 
\begin{align}
\frac{dS_{j}}{dt}&
=F_{j}(S,T)=r_{j-1}f_{j-1}S_{j-1}-a_{j}S_{j}-f_{j}S_{j}-p_{j}S_{j}T_{j}
\label{eq:OriginalModel} \\
\frac{dT_{j}}{dt}& =G_{j}(S,T)=c_{j}S_{j}T_{j}-bT_{j}.  \notag
\end{align}
Here $S_{j}$ denotes the pathogen population at stage $j$, $T_{j}$ is the
cognate host response. The indices $j=0,\dots ,n-1$\ are taken modulo $n$.
The parameters represent the following processes:

\begin{itemize}
\item $a_{j}$ is the decay rate of stage $S_{j}$. If $a_{j}$ is negative,
this state proliferates.

\item $f_{j}$ is the rate at which stage $S_{j}$ is lost to become (or
produce) stage $S_{j+1}$.

\item $r_{j}$ is an amplification factor in the process by which stage $
S_{j} $ becomes (or produces) stage $S_{j+1}$. For example, the loss of one
lytically infected cell may produce $r_{j}\cong 10^{4}$ free virus.

\item $p_{j}$ represents the efficacy of the immune response $T_{j}$ in
killing infected stage $S_{j}$.

\item $c_{j}$ is the antigenicity of stage $S_{j}$, i.e., its efficacy in
inducing proliferation of immune response $T_{j}$.

\item $b$ is the natural death rate of the response $T_{j}$. We assume it is
the same for all stages.
\end{itemize}

We refer to the parameters collectively as $\theta $. Except for $a_{j}$, $
j=0,\dots ,n-1$, these are assumed non-negative.

Our flagship result is that while (\ref{eq:OriginalModel}) has $2^{n}$ fixed
points for generic values of $\theta $, exactly one of these is biologically
meaningful and stable (\cite{cyclicPathogen}).

Let us focus for the moment on the terms $-p_{j}S_{j}T_{j}$ of $F_{j}$ and $
c_{j}S_{j}T_{j}$ of $G_{j}$. The term $-p_{j}S_{j}T_{j}$ represents the
killing of pathogen at stage $j$ (usually infected cells in a particular
differentiation state) by the cognate T-cell population. This takes place
pursuant to an encounter between T-cells and infected cells displaying antigen
complexed to MHC. To a first approximation the rate of such encounters is
proportional to the product of the sizes of the two populations. Thus, to a
first approximation, this term reflects the mechanism of the biological
process it represents.

The term $c_{j}S_{j}T_{j}$ of $G_{j}$ represents proliferation of
T-cells in response to the presence of antigen. Here, the underlying
biological processes are considerably more complicated, involving a
number of cell types. Initially, T-cells are activated and begin to
proliferate only in response to antigen presenting cells, particularly
dendritic cells (DCs) (\cite{janeway}).  The density of presented
antigen is known to affect these T-cell-DC interactions
(\cite{Henrickson}, \cite{Zheng}), eliciting differing CTL responses at
different densities, including T-cell exhaustion at high
concentrations of presented antigen (\cite{Mueller}).  Activated CD8+
T-cells also exhibit central and effector memory phenotypes and the
relationships between these phenotypes is not well understood
(\cite{generationOfEffector}).  Finally, the length of the cell cycle
places a hard limit on the rate at which the T-cell population can
proliferate.  Thus, the rate of T-cell proliferation becomes
\emph{saturated} for large amounts of antigen (\cite{Hudriser}). In
this, they bear a similarity to the rates of enzyme catalyzed chemical
reactions (reviewed in \cite {ReviewOnModelingBiochemicalReactions},
see \cite{EnzymeCatalyzedReactions} for experimental evidence).

To accommodate dose-dependent effects, we will study the system 
\begin{align}
\frac{dS_{j}}{dt}& =\hat{F}
_{j}(S,T)=r_{j-1}f_{j-1}S_{j-1}-a_{j}S_{j}-f_{j}S_{j}-p_{j}S_{j}T_{j}
\label{eq:ModelWithSaturation} \\
\frac{dT_{j}}{dt}& =\hat{G}_{j}(S,T)=\phi _{j}(S_{j})T_{j}-bT_{j}.  \notag
\end{align}
which generalizes (\ref{eq:OriginalModel}).  The function $\hat H$ is
unchanged from $H$.  The terms of $\hat{G}_{j}$ represent proliferation of
CTLs in response to the presence of antigen and the loss of CTLs due to death
or decommissioning. We use functions $\phi _{j}:[0,\infty)\to[0,\infty)$ 
 to denote the dose-response curves. We assume that for each $j$, $\phi
_{j}(0)=0$, and that for $x\in (0,\infty )$, $\phi _{j}^{\prime }(x)$ exists
and is positive. In particular, each $\phi _{j}$ is continuous on $[0,\infty
)$ and strictly monotone increasing. The possibility of dose-response
saturation arises from the case where there is an $m_{j}\in \mathbb{R}$ so
that $\lim_{x\rightarrow \infty }\phi _{j}(x)=m_{j}$. We show that with 
appropriate modification, the major results of \cite{cyclicPathogen} hold
for (\ref{eq:ModelWithSaturation}). While the term $\phi _{j}(S_{j})$ is
still phenomenological in that it omits discussion of biological mechanism,
we argue in Section~\ref{sec:discussion} that it may well offer a way to
address this limitation.

The system (\ref{eq:ModelWithSaturation}) can exhibit a behavior which does
not arise with (\ref{eq:OriginalModel}). In either case if $a_{j}+f_{j}<0$,
we say that $j$ is \emph{self-establishing}. It is not hard to see that a
self-establishing stage which is not regulated will expand without bound. On
the other hand, as we will show, if $m_{j}\leq b$, the host cannot mount a
response to $S_{j}$. It is \emph{\ immuno-incompetent} with respect to this
stage. If $j$ is self-establishing and the host is immuno-incompetent with
respect to $j$, we say that the parameter set is \emph{fatal}. If the host
is immunologically incompetent at all stages we say the host is \emph{
totally immunologically incompetent}. Clearly, in this case, if the basic
reproductive number of the pathogen is greater than one, infection is also
fatal to the host. Accordingly, we will assume that the host is
immuno-competent for at least one stage.

\section{Background and definitions}

\label{sec:background}

We start by transforming (\ref{eq:ModelWithSaturation}) through a change of
coordinates. For this purpose we take $m_{j}:=\lim_{x\rightarrow \infty
}\phi _{j}(x)\in 
\mathbb{R}
\cup \{\infty \}$. Notice that if $m_{j}>b$, there is a unique value $
b_{j}\in 
\mathbb{R}
$ so that $\phi _{j}(b_{j})=b$. If $m_{j}\leq b$, we take $b_{j}:=\infty $.
So we define 
\begin{equation*}
c_{j}:=
\begin{cases}
\phi _{j}^{\prime }(b_{j}) & \text{if $b_{j}<\infty $} \\ 
1 & \text{otherwise}
\end{cases}
\end{equation*}
In the former case, $c_{j}$ is the marginal antigenicity of $S_{j}$ at the
value $b_{j}$. We now use the linear change of coordinates 
\begin{align*}
H& :
\mathbb{R}
^{2n}\rightarrow 
\mathbb{R}
^{2n} \\
(S_{j},T_{j})& \mapsto (\bar{S}_{j},\bar{T}
_{j}):=H_{j}(S_{j},T_{j}):=(c_{j}S_{j},p_{j}T_{j}).
\end{align*}
This gives the equations 
\begin{align}
\frac{d\bar{S}_{j}}{dt}& =\bar{\hat{F}_{j}}(\bar{S},\bar{T})=\bar{r}
_{j-1}f_{j-1}\bar{S}_{j-1}-a_{j}\bar{S}_{j}-f_{j}\bar{S}_{j}-\bar{S}_{j}\bar{
T}_{j}  \label{eq:workingSaturationModel} \\
\frac{d\bar{T}_{j}}{dt}& =\bar{\hat{G}_{j}}(\bar{S},\bar{T})=\bar{\phi}_{j}(
\bar{S}_{j})\bar{T}_{j}-b\bar{T}_{j}  \notag \\
\bar{r}_{j}& =\frac{c_{j+1}}{c_{j}}r_{j}  \notag \\
\bar{\phi}_{j}(\bar{S}_{j})& =\phi _{j}\left( \frac{S_{j}}{c_{j}}\right)  
\notag
\end{align}
Note that for each $j$, $\bar{\phi}_{j}$ still enjoys the properties that it
is differentiable, $\bar{\phi}_{j}^{\prime }(x)>0$ for $x>0$ and $\bar{\phi}
_{j}(0)=0$. We now take $\bar{b}_{j}$ to be the unique solution to $\bar{\phi
}_{j}(\bar{b}_{j})=b$, i.e., $\bar{b}_{j}:=$ $c_{j}b_{j}$, where such
exists. Note that $\bar{\phi}_{j}$ now enjoys the additional property that $
\bar{\phi}_{j}^{\prime }(\bar{b}_{j})=1$. In the case studied in \cite
{cyclicPathogen}, $\phi _{j}(S_{j})=c_{j}S_{j}$ and thus $\bar{\phi}
_{j}(S_{j})=S_{j}$, giving 
\begin{align}
\frac{d\bar{S}_{j}}{dt}& =\bar{\hat{F}_{j}}(\bar{S},\bar{T})=\bar{r}
_{j-1}f_{j-1}\bar{S}_{j-1}-a_{j}\bar{S}_{j}-f_{j}\bar{S}_{j}-\bar{S}_{j}\bar{
T}_{j}  \label{eq:workingOriginalModel} \\
\frac{d\bar{T}_{j}}{dt}& =\bar{\hat{G}_{j}}(\bar{S},\bar{T})=\bar{S}_{j}\bar{
T}_{j}-b\bar{T}_{j}  \notag
\end{align}
We will henceforth drop the bars and assume our equations are given in the
form (\ref{eq:workingSaturationModel}). We take a parameter set $\theta $ to
be a set of values for $b$, $r_{j}$, $f_{j}$, $a_{j}$, $m_{j}$ and $b_{j},$ $
j=0,...,n-1$. When we need to make an explicit comparison with the $\phi _{j}
$ of (\ref{eq:ModelWithSaturation}), we will refer to the later as
\textquotedblleft biological $\phi $\textquotedblright , $\phi _{j}^{\text{
bio}}$.

We will adopt the following notational conventions. Sets such $[j,k]$ and $
[j,k)$ are to be taken cyclically. That is to say, if $j<k$, then $
[j,k]=\{j,\dots ,k\}$, while if $j>k$, $[j,k]=\{j,\dots ,n-1,0,\dots ,k\}$.
We take $[j,j)$ to be the empty set so that any product taken over $[j,j)$
is equal to one. We abuse notation by taking $[0,n)=\{0,\dots ,n-1\}$ 

We now review and in some cases generalize the definitions of \cite
{cyclicPathogen}.

\begin{mydef}
Given a fixed point $(S^*,T^*)$ of (\ref{eq:workingSaturationModel}), the 
\emph{regulated} and \emph{\ unregulated} stages of $(S^*,T^*)$ are 
\begin{align*}
\reg{(S^*,T^*)} &= \{j \mid T^*_j \ne 0\} \\
\unreg{(S^*,T^*)} &= \{j \mid T^*_j = 0\}
\end{align*}
\end{mydef}

\begin{mydef}
$(S^{\ast },T^{\ast })$ is \emph{biologically meaningful} if $S_{j}^{\ast
}\geq 0$, $T_{j}^{\ast }\geq 0$ for $j=0,\dots ,n-1$. It is \emph{infected}
if for some (hence, all, see 1) below) $j$, $S_{j}^{\ast }>0$.
\end{mydef}

\begin{mydef}
Given a parameter set $\theta$, the \emph{self-establishing} stages are 
\begin{equation*}
\text{SE}(\theta) = \{j \mid a_j + f_j < 0 \}.
\end{equation*}
\end{mydef}

\begin{mydef}
The \emph{immuno-incompetent} stages of $\theta $ are 
\begin{equation*}
\text{Incomp}(\theta )=\{j\mid m_{j}\leq b\}=\{j\mid b_{j}=\infty \}.
\end{equation*}
\end{mydef}

\begin{mydef}
If $j\in\text{SE}(\theta)\cap\text{Incomp}(\theta)$, we say that the stage $j
$ and the parameter set $\theta$ are \emph{fatal}. We will assume that the
host is capable of mounting a response to at least one stage, i.e., $\text{
Incomp}(\theta) \ne [0,n)$.
\end{mydef}

\begin{mydef}
\label{def:followOn} If $\text{SE}(\theta )=\emptyset $, the \emph{follow-on}
constants of $\theta $ are 
\begin{align*}
M_{j}& =\frac{r_{j}f_{j}}{a_{j+1}+f_{j+1}} \\
M_{jk}& =\prod_{\ell \in \lbrack j,k)}M_{\ell }
\end{align*}
In the case where $\text{SE}(\theta )\neq \emptyset $, $M_{j}$ is only \emph{
meaningful} for our purposes for $j+1\notin \text{SE}(\theta )$.
Accordingly, $M_{jk}$ is only \emph{meaningful} if $(j,k]\cap \text{SE}
(\theta )=\emptyset $. Note that for every $k\in \lbrack j,\ell )$ it holds $
M_{j\ell }=M_{jk}M_{k\ell }$.
\end{mydef}

\begin{mydef}
\label{def:starves} We say that $j$ \emph{starves} $k$ and write $j\starves k
$ if $b_{j}M_{jk}<b_{k}$. Here we assume that $M_{jk}$ is meaningful and
that $b_{j}$ is finite, though $b_{k}$ need not be. In particular, if $M_{jk}
$ is meaningful, $j\notin \text{Incomp}(\theta )$, and $k\in \text{Incomp}
(\theta )$, then $j\starves k$.
\end{mydef}

\begin{mydef}
The \emph{starvable} stages of $\theta$ are 
\begin{equation*}
\str{(\theta)} = \{k \mid \text{there is $j$ so that $j\starves k$} \}. 
\end{equation*}
The \emph{unstarvable} stages $\unstr{(\theta)}$ are the complement of
these.
\end{mydef}

\begin{mydef}
A biologically meaningful fixed point $(S^{\ast },T^{\ast })$ is \emph{
saturated}\footnote{
There is an unfortunate collision here between the use of the term saturated
to denote the host mounting a T-response to all stages capable of supporting
one (\cite{cyclicPathogen}) and the meaning of the
term used in the Introduction above, namely, a maximum prolfieration
rate, with no increase through further stimulation.}
 if $\reg{(S^{\ast },T^{\ast })}=\unstr{(\theta )}$. It is
 \emph{moderated} if for $j\in \text{Unreg}$, $S_{j}^{\ast }<b_{j}$.
\end{mydef}

\begin{mydef}
If $\text{SE}(\theta)=\emptyset$, we define 
\begin{equation*}
R_0 = \prod_{j=0}^{n-1} M_j .
\end{equation*}
$R_0$ may be interpreted as the number of copies of the pathogen produced by
a single copy entering a naive host (\cite{cyclicPathogen}). It is not hard to
see that $R_0$ is invariant under the transformation $H$ as befits a
property of the organism being described.
\end{mydef}

\begin{mydef}
When we say that $\theta$ is \emph{generic} we will require the following:

\begin{itemize}
\item $R_0 \ne 1$.

\item There is no $j$ so that $a_j+f_j = 0$.

\item There is no pair $(j,k)$ so that $b_{j}<\infty $, $b_{k}<\infty $ and $
b_{j}M_{jk}=b_{k}$.

\item At the saturated biologically meaningful fixed point, there are $j$
and $k$ so that $T_{j}^{\ast }\neq T_{k}^{\ast }$.
\end{itemize}

It is not hard to see that each of these conditions has measure zero, thus
justifying the use of the term generic. The detailed motivation for these
exclusions can be found in \cite{cyclicPathogen}.
\end{mydef}

\begin{mydef}
Suppose $\reg{(S^*,T^*)}\ne\emptyset$. Given a stage  $k$, we define $h_k
$ to be the unique stage such that $h_k\in  \reg{(S^*,T^*)}$ and $
(h_k,k) \subset\text{Unreg}(S^*,T^*)$.
\end{mydef}

\section{Results}

\label{sec:results}

We will start by assuming that $\text{SE}(\theta )=\emptyset $. This will be
a standing assumption until it is lifted in Section~\ref{sec:SE}.

The linear stability analysis performed in \cite{cyclicPathogen} is possible
because we were able to calculate the characteristic polynomial of the
Jacobian matrix of the right hand side of (\ref{eq:workingOriginalModel}),
which corresponds to setting $\phi _{j}(S_{j})=S_{j}$ for each $j$ in (\ref
{eq:workingSaturationModel}) (recall that we are omitting the bars). Here we
contemplate more general $\phi _{j}:
\mathbb{R}
\rightarrow 
\mathbb{R}
,$ $j=0,...,n-1$ (which is a consequence of contemplating more general $\phi
_{j}^{\text{bio}}:
\mathbb{R}
\rightarrow 
\mathbb{R}
$) with the properties mentioned above. Consequently, in order to make use
of the results obtained in \cite{cyclicPathogen}, we need to establish what
changes are induced on the Jacobian matrix through the use of more general
functions $\phi _{j}$. The partial derivatives of the right hand side of (
\ref{eq:workingSaturationModel}) are given by 
\begin{align*}
\frac{\partial \hat{F}_{k}}{\partial S_{j}}& =
\begin{cases}
r_{k-1}f_{k-1} & \text{if $j=k-1$} \\ 
-a_{k}-f_{k}-T_{k} & \text{if $j=k$} \\ 
0 & \text{otherwise}
\end{cases}
\\
\frac{\partial \hat{F}_{k}}{\partial T_{j}}& =
\begin{cases}
-S_{k} & \text{if $j=k$} \\ 
0 & \text{otherwise}
\end{cases}
\\
\frac{\partial \hat{G}_{k}}{\partial S_{j}}& =
\begin{cases}
\phi _{k}^{\prime }(S_{k})T_{k} & \text{if $j=k$} \\ 
0 & \text{otherwise}
\end{cases}
\\
\frac{\partial \hat{G}_{k}}{\partial T_{j}}& =
\begin{cases}
\phi _{k}(S_{k})-b & \text{if $j=k$} \\ 
0 & \text{otherwise}
\end{cases}
\end{align*}
Since the functions $\hat{F}_{j}$ do not depend on any $\phi _{j}$, only the
partial derivatives $\frac{\partial \hat{G}_{k}}{\partial S_{j}}$ and $\frac{
\partial \hat{G}_{k}}{\partial T_{j}}$ differ from the results obtained in 
\cite{cyclicPathogen}. As we shall see, most differences vanish
when the functions are evaluated at a fixed point.

\begin{proposition}
\label{prop:R0} ~

\begin{enumerate}
\item $R_0$ is the basic reproductive number of the pathogen.

\item If $R_0<1$, the pathogen fails to establish infection and $(S^*,T^*) =
0$ is  a local attractor. If $R_0>1$, the pathogen is able to establish 
infection. In particular, this makes $(S^*,T^*) =0$ an unstable fixed point.

\item If $R_{0}<1$, $(S^{\ast },T^{\ast })=0$ is a global attractor.
\end{enumerate}
\end{proposition}

\begin{proof}
There are two ways to establish the first and second claims.  One is
by using the interpretation of $R_{0}$ in terms of the lifespan and
productivity of each stage.  The other is by computing the eigenvalues
the Jacobian.  We briefly sketch the first approach.  In the absence
of immune response, stage 0 has an expected lifespan of
$\frac{1}{a_0+f_0}$.  During that time, it produces
$\frac{r_0f_0}{a_0+f_0}$ copies of stage 1.  These, in turn, produce
$\frac{r_0f_0}{a_0+f_0} \frac{r_1f_1}{a_1+f_1}$ copies of stage 2.
Continuing in this way produces $R_0$ copies of stage 0.  If this is
greater than 1, the pathogen can establish infection, if less than 1,
not.

To see these two claims using the eigenvalues of the Jacobian, notice
that since $\phi _{j}(0)=0$, $ j=0,...,n-1$, the Jacobian matrix
evaluated at $(S^{\ast },T^{\ast })=(\vec{0 },\vec{0})$ is identical
to the one obtained in \cite[Proposition 1] {cyclicPathogen}.  Thus
the claim follows from Propositions~1 and 2 of \cite{cyclicPathogen}. 

The third claim comes from showing that the reproductive number in the
presence of immune response is no more than the reproductive number in
the naive host as in \cite[Proposition 1]{cyclicPathogen}. \qed
\end{proof}

The following correspond to the numbered observations in Section 3 of \cite{cyclicPathogen} and follow from the fixed point equations 
\begin{align*}
\dot{S_{j}}& =\hat{F}_{j}(S^{\ast },T^{\ast })=r_{j-1}f_{j-1}S_{j-1}^{\ast
}-S_{j}^{\ast }(a_{j}+f_{j}+T_{j}^{\ast })=0 \\
\dot{T_{j}}& =\hat{G}_{j}(S^{\ast },T^{\ast })=(\phi _{j}(S_{j}^{\ast
})-b)T_{j}^{\ast }=0
\end{align*}

\noindent\textbf{1)} Given $(S^*,T^*)$, if there is $j$ such that $S_j^* = 0$
then $(S^*,T^*) = 0$.

\noindent\textbf{2)} If $j\in\text{Reg}(S^*,T^*)$, then $S_j^* = b_j$.

\noindent\textbf{3)} If $j\in\text{Unreg}(S^*,T^*)$, then $T_j^*=0$.

\noindent\textbf{4)} If $j+1\in\text{Unreg}(S^*,T^*)$, then $S_{j+1}^*=S_j^*
M_j$.

\noindent \textbf{5)} If $[j+1,k]\subset \text{Unreg}(S^{\ast },T^{\ast })$
then $S_{k}^{\ast }=S_{j}^{\ast }M_{jk}$. This follows by induction on the
previous observation.

\noindent \textbf{6)} Assume Reg$(S^{\ast },T^{\ast })\neq \emptyset $. If $
k\in \text{Unreg}(S^{\ast },T^{\ast })$, then $S_{k}^{\ast
}=b_{h_{k}}M_{h_{k}k}$. This follows from 2) and 5).

\noindent \textbf{7)} If $\theta $ is generic and $(S^{\ast },T^{\ast })\neq
0$ then $\text{Reg}(S^{\ast },T^{\ast })\neq \emptyset $. Were Reg$(S^{\ast
},T^{\ast })=\emptyset $, then by 5) $S_{0}^{\ast }=S_{0}^{\ast }R_{0}$.
Consequently $R_{0}=1$, contradicting our first genericity requirement.

\noindent\textbf{8)} If $j\in\text{Reg}(S^*,T^*)$ then $T_j^* = \frac{
r_{j-1}f_{j-1}}{b_j} S_{j-1}^* -(a_j+f_j)$.

\noindent \textbf{9)} If $j\in \text{Reg}(S^{\ast },T^{\ast })$ then $
T_{j}^{\ast }=r_{j-1}f_{j-1}\frac{b_{h_{j}}}{b_{j}}M_{h_{j}j-1}-(a_{j}+f_{j})
$. This follows from the fact that $S_{j-1}^{\ast }=S_{h_{j}}^{\ast }M_{{
h_{j}}{j-1}}$ (which follows from 6) if $j-1\in $Unreg$(S^{\ast },T^{\ast })$
, and holds trivially, if $j-1\in $Reg$(S^{\ast },T^{\ast })$) and $
S_{h_{j}}^{\ast }=b_{h_{j}}$.

\noindent \textbf{10)} If $j\in \text{Reg}(S^{\ast },T^{\ast })$, then $
T_{j}^{\ast }>0$ if and only if $b_{h_{j}}M_{h_{j}j}>b_{j}$. This follows
from the previous observation; (recall that we have assumed SE$(\theta
)=\emptyset $). 

\begin{proposition}
\label{prop:starves} Suppose $\theta $ is generic and $(S^{\ast },T^{\ast })$
is a biologically meaningful fixed point. Suppose further that $j\starves k$
. If $j\in \text{Reg}(S^{\ast },T^{\ast })$, then $k\in \text{Unreg}(S^{\ast
},T^{\ast })$.
\end{proposition}

\begin{proof}
Suppose $j\succ k$ and $j\in $Reg$(S^{\ast },T^{\ast })$. Let $
\{j=j_{0},j_{1},\dots ,j_{m}\}=[j,k)\cap $Reg$(S^{\ast },T^{\ast })$
(cyclically ordered as listed) and let $j_{m+1}=k$. If $m=0$, then, $h_{k}=j$
and, due to $j\succ k$, it holds $b_{h_{k}}M_{h_{k}k}<b_{k}$. Thus, the
claim follows by observation 10), above. Otherwise we have $j\neq j_{m}=h_{k}
$ and it suffices to show that $S_{j_{m}}^{\ast }M_{j_{m}k}<b_{k}$. For $
\ell =0,\dots ,m$, we have $j_{\ell }\in $Reg$(S^{\ast },T^{\ast })$ so we
must have $S_{j_{\ell }}^{\ast }=b_{j_{\ell }}$. For $\ell =0,\dots ,m-1$ we
must also have $S_{j_{\ell +1}}^{\ast }<S_{j_{\ell }}^{\ast }M_{j_{\ell
}j_{\ell +1}}$, because $h_{j_{\ell +1}}=j_{\ell }$. We then have 
\begin{align*}
b_{j_{1}}& =S_{j_{1}}^{\ast }<S_{j_{0}}^{\ast }M_{j_{0}j_{1}} \\
b_{j_{2}}& =S_{j_{2}}^{\ast }<S_{j_{1}}^{\ast
}M_{j_{1}j_{2}}<S_{j_{0}}^{\ast }M_{j_{0}j_{2}} \\
& \vdots  \\
b_{j_{m}}& =S_{j_{m}}^{\ast }<S_{j_{0}}^{\ast
}M_{j_{0}j_{m}}=b_{j_{0}}M_{j_{0}j_{m}}
\end{align*}
so that 
\begin{equation*}
S_{h_{k}}^{\ast }M_{h_{k}k}=S_{j_{m}}^{\ast
}M_{j_{m}k}<b_{j_{0}}M_{j_{0}k}<b_{k}.
\end{equation*}
Thus $k\in $Unreg$(S^{\ast },T^{\ast })$ as required. \qed
\end{proof}

\begin{proposition}
Let $\theta $ be a generic parameter set such that $R_{0}>1$. Then $\starves$
is a strict partial order.
\end{proposition}

\begin{proof}
We must show that $\succ$ is anti-reflexive, asymmetric and transitive. The
first follows immediately from the fact that $M_{jj}=1$.

To see that $\succ $ is asymmetric, suppose we have $j\succ k$ and $k\succ j$
. We then have $b_{j}M_{jk}<b_{k}$ and $b_{k}M_{kj}<b_{j}$. This gives $
b_{j}>b_{j}M_{jk}M_{kj}$. But $M_{jk}M_{kj}=R_{0}$, contradicting $R_{0}>1$.

To see the third we suppose that $j\succ k$ and $k\succ \ell $. We consider
two cases, $k\in \lbrack j,\ell ]$ and $\ell \in \lbrack j,k]$. In the first
case, we have $M_{jk}M_{k\ell }=M_{j\ell }$. We then have $b_{j}M_{jk}<b_{k}$
, $b_{k}M_{k\ell }<b_{\ell }$ giving $b_{j}M_{j\ell }=b_{j}M_{jk}M_{k\ell
}<b_{k}M_{k\ell }<b_{\ell }$ as required. In the second case, we have $
M_{jk}M_{k\ell }=R_{0}M_{j\ell }$, so that $b_{j}R_{0}M_{j\ell
}=b_{j}M_{jk}M_{k\ell }<b_{k}M_{k\ell }<b_{\ell }$. Since $R_{0}>1$, this
implies $j\succ \ell $ as required. \qed
\end{proof}

\begin{myremark}
Since $\starves$ is a partial order, it is cycle-free, that is there is no
sequence of stages $j_{0}\starves j_{j}\starves\dots \starves j_{0}$.
Consequently, we can define the \emph{depth} of a stage $k$, $d(k)$ to be
the length of the longest chain $j_{0}\starves\dots \starves j_{d(k)}=k$. It
follows that $\unstr{(\theta )}$ consists of the stages of depth 0. In
particular, $\unstr{(\theta )}\neq \emptyset $. Note that if $\theta $
is such that no two stages are comparable, then $\succ $ is empty and every
stage is $\succ $-maximal, so $\unstr{(\theta )}=[0,n)$. $\str{(\theta )}$
consists of the stages of positive depth. If $\text{Incomp}(\theta )\neq
\emptyset $, $\text{Incomp}(\theta )$ consists of the stages of maximal
depth.
\end{myremark}

\begin{proposition}
\label{prop:saturatedEqualsModerated}Suppose that $\theta $ is generic.
Furthermore, let $(S^{\ast },T^{\ast })$ be a biologically meaningful
infected fixed point. Then the pathogen populations are moderated at $
(S^{\ast },T^{\ast })$ if and only if the immune response is saturated at $
(S^{\ast },T^{\ast })$.
\end{proposition}

\begin{proof}
We first show that if $(S^{\ast },T^{\ast })$ is moderated, then $(S^{\ast
},T^{\ast })$ is saturated.

We claim $\unstr{(\theta )}{\subseteq }$\reg{$
(S^{\ast },T^{\ast })$}. Suppose to the contrary $j\in $$\unstr{
(\theta )\cap }$\unreg{$(S^{\ast },T^{\ast })$}. By assumption $
(S^{\ast },T^{\ast })$ is moderated, so $S_{j}^{\ast }<b_{j}$. Since $j\in $
\unreg{$(S^{\ast },T^{\ast })$}, by 7), 6) and 2) above, $
S_{j}^{\ast }=S_{h_{j}}^{\ast }M_{h_{j}j}=b_{h_{j}}M_{h_{j}j}$. This gives $
b_{h_{j}}M_{h_{j}j}<b_{j}$, i.e., $h_{j}\succ j$, contradicting the
assumption that $j\in \unstr{(\theta )}$. This proves the claim.

We claim that $\str(\theta )\subseteq $Unreg$(S^{\ast },T^{\ast })$. If $\str
(\theta )=\emptyset $, this holds trivially. Suppose $k\in \str(\theta )$.
Then there is a maximal $j$ so that $j\succ k$. Being maximal $j\in \unstr
{(\theta )}$ and thus $j\in $Reg$(S^{\ast },T^{\ast })$. It
follows by Proposition~\ref{prop:starves} that $k\in $Unreg$(S^{\ast
},T^{\ast })$ as required.

We now show that if $(S^{\ast },T^{\ast })$ is saturated, then $(S^{\ast
},T^{\ast })$ is moderated.

If $\mathop{\textup
{Unreg}}(S^{\ast },T^{\ast })=\emptyset ,$ the claim holds vacuously.
Suppose that $k\in $Unreg$(S^{\ast },T^{\ast })$. We must show that $
S_{k}^{\ast }<b_{k}$. Again, we choose $j$ to be maximal so that $j\succ k$.
Since $(S^{\ast },T^{\ast })$ is saturated, $j\in $Reg$(S^{\ast },T^{\ast })$
, thus, by 2) $S_{j}^{\ast }=b_{j}$. If $[j+1,k)\subseteq 
\mathop{\textup{Unreg}}(S^{\ast },T^{\ast })$, we are done, for then $\S 
_{k}=\S _{j}M_{jk}=b_{j}M_{jk}<b_{k}$. On the other hand if $[j+1,k)\cap 
\mathop{\textup{Reg}}(S^{\ast },T^{\ast })\neq \emptyset $, choose $m\in
\lbrack j+1,k)\cap \mathop{\textup {Reg}}(S^{\ast },T^{\ast })$ so that $
m=h_{k}$. Since $m\in \mathop{\textup {Reg}}(S^{\ast },T^{\ast })$, by the
assumed saturation $m\in \unstr{(\theta )}$. Therefore $j\nsucc m$, in other
words, $b_{j}M_{jm}\geq b_{m}$. If $b_{m}M_{mk}\geq b_{k}$, these two
inequalities would yield $b_{j}M_{jk}=b_{j}M_{jm}M_{mk}\geq b_{k}$
contradicting $j\succ k$. Consequently $b_{m}M_{mk}<b_{k}$ must hold. Now we
have $\S _{k}=\S _{m}M_{mk}=b_{m}M_{mk}<b_{k}$ as required. \qed
\end{proof}

\begin{theorem}
\label{thm:unsaturatedUnstable} Suppose $\theta $ is generic and $(S^{\ast
},T^{\ast })$ is a biologically meaningful infected fixed point which is not
saturated. Then there is $j\in \text{Unreg}(S^{\ast },T^{\ast })$ so that
for any open neighborhood $U$ of $(S^{\ast },T^{\ast })$ there is a
biologically meaningful point $x\in U$ so that $\frac{dT_{j}}{dt}\big|_{x}>0$
. In particular $(S^{\ast },T^{\ast })$ is unstable.
\end{theorem}

\begin{proof}
Since $(S^{\ast },T^{\ast })$ is not saturated, it is not moderated.  Thus,
there is $j\in \unreg(S^{\ast },T^{\ast })$ with $S_{j}^{\ast } \ge b_{j}$.
Since $\theta$ is generic, $S_{j}^{\ast } > b_{j}$, for otherwise we would
have $b_{h_j}M_{h_j j} = b_j$.  In particular, $b_{j}<\infty$.  It follows
that $j\notin \incomp(\theta )$.  Since $S_{j}^{\ast}>b_{j}$, $\phi
_{j}(S_{j}^{\ast })>b$, so $\frac{\partial \hat{G}_{j}}{\partial
  T_{j}}|_{(S^{\ast },T^{\ast })}=\phi _{j}(S_{j})-b>0$.

Let $e_{T_{j}}$ be the unit vector in the $T_{j}$ direction. Then, for
any $\delta >0$, $\dot{T_{j}}|_{(S^{\ast },T^{\ast })+\delta
  e_{T_{j}}}>0$. Thus, in any open neighborhood $U$ of $(S^{\ast
},T^{\ast })$, there are biologically meaningful points whose orbits
move away from $(S^{\ast },T^{\ast })$. In particular, $(S^{\ast
},T^{\ast })$ is unstable as required. \qed
\end{proof}

\begin{theorem}
\label{thm:uniqueStableFixedPoint} Suppose that $\theta$ is generic and that 
$(S^*,T^*)$ is a biologically meaningful infected fixed point. In
particular, not all $T_j^*$ are equal. If $(S^*,T^*)$ is moderated then $
(S^*,T^*)$ is a local asymptotically stable equilibrium. In particular, the
eigenvalues of the Jacobian matrix $J(S^*,T^*)$ have strictly negative real
part.
\end{theorem}

\begin{corollary}
\label{cor:uniqueStableFixedPoint} For a generic parameter set, the system (
\ref{eq:workingSaturationModel}) (and hence (\ref{eq:ModelWithSaturation}))
has a unique biologically meaningful stable fixed point.
\end{corollary}

\begin{proof}
Since the sets of starvable and unstarvable stages depend only on
$\theta$, there is exactly one saturated fixed point, hence exactly
one moderated fixed point.  The Corollary now follows from the Theorem.
\qed
\end{proof}

\begin{proof}[Theorem \protect\ref{thm:uniqueStableFixedPoint}]
The proof of the corresponding Theorem in \cite{cyclicPathogen} proceeds by
showing that the Jacobian matrix of the system (\ref{eq:workingOriginalModel}
) has eigenvalues all of whose real parts are negative. It will therefore
suffice to show that we can carry out the same computation on the Jacobian
matrix of (\ref{eq:workingSaturationModel}) evaluated at a moderated
fixed point $\fp$.  Since $H$ and $\hat H$ are identical, we need only
consider the partials of $G$ and $\hat G$.  We have 
\begin{align*}
\frac{\partial G_{k}}{\partial S_{j}}& =
\begin{cases}
T_{k} & \text{if $j=k$} \\ 
0 & \text{otherwise}
\end{cases}
\\
\frac{\partial \hat{G}_{k}}{\partial S_{j}}& =
\begin{cases}
\phi _{k}^{\prime }(S_{k})T_{k} & \text{if $j=k$} \\ 
0 & \text{otherwise}
\end{cases}
\end{align*}
Now for $k\in \mathop{\textup {Unreg}}(S^{\ast },T^{\ast })$, both of these
partial derivatives vanish, while if $k\in \mathop{\textup {Reg}}(S^{\ast
},T^{\ast })$, We then have $\S_{k}=b_{k}$ so that $\phi _{k}^{\prime
}(\S_{k})=1$, and once again, the two are identical.

Moreover, we have 
\begin{align*}
\frac{\partial G_{k}}{\partial T_{j}}& =
\begin{cases}
S_{k}-b & \text{if $j=k$} \\ 
0 & \text{otherwise}
\end{cases}
\\
\frac{\partial \hat{G}_{k}}{\partial T_{j}}& =
\begin{cases}
\phi _{k}(S_{k})-b & \text{if $j=k$} \\ 
0 & \text{otherwise}
\end{cases}
\end{align*}
Now if $k\in \mathop{\textup {Reg}}(S^{\ast },T^{\ast })$, we have $\S_{k}=b$
so that $S_{k}-b=0$ in the former case, while in the latter case we have $
\S_{k}=b_{k}$ so that $\phi _{k}(\S_{k})-b=0$. Finally, in the case where $
k\in \mathop{\textup {Unreg}}(S^{\ast },T^{\ast })$, the proof of \cite[
Theorem 2]{cyclicPathogen} appeals to the fact that $(S^{\ast },T^{\ast })$
is moderated, thus ensuring that $\S_{k}-b<0$. Here, the fact that $(S^{\ast
},T^{\ast })$ is moderated implies that $\S_{k}<b_{k}$ so that $\phi
_{k}(\S_{k})-b<0$ and we can proceed as before. \qed
\end{proof}

\subsection{Self-establishing stages}

\label{sec:SE} We now turn to the case where $\text{SE}(\theta)\ne\emptyset$
. In this case we need the assumption that $\theta$ is not fatal, that is, $
\text{SE}(\theta) \cap \text{Incomp}(\theta)=\emptyset$ and $\text{Incomp}
(\theta) \ne [0,n)$.

We start by observing that if $\text{SE}(\theta )\neq \emptyset $, then the
pathogen is viable. Accordingly, in place of Proposition~\ref{prop:R0}, we
have the following.

\begin{proposition}
If $\text{SE}(\theta )\neq \emptyset $, then $(S^{\ast },T^{\ast })=(0,0)$
is an unstable equilibrium. In particular, the pathogen is able to infect
the host.
\end{proposition}

\begin{proof}
Suppose that $j\in {\mathop{\textup {SE}}}(\theta )$. Then $\frac{\partial \hat{F}_{j}}{\partial S_{j}}|_{(0,0)}=-a_{j}-f_{j}>0$. This gives orbits with positive and
increasing $S_{j}$ inside any open set around $(0,0)$. \qed
\end{proof}

The numbered observations 1) through 9) listed above hold without change.
Observation 10) now requires the additional hypothesis that
$j\notin\se(\theta)$, giving
\medskip

\noindent \textbf{10$'$)} If $j\in \text{Reg}(S^{\ast },T^{\ast })$ and
$j\notin\se(\theta)$, then $ T_{j}^{\ast }>0$ if and only if
$b_{h_{j}}M_{h_{j}j}>b_{j}$. 
\medskip

As before, this follows from observation 9).

\begin{proposition}
\label{prop: SelfEstablishmentImpliesRegulation} Suppose that $\text{SE}
(\theta )\neq \emptyset $ and $(S^{\ast },T^{\ast })$ is a biologically
meaningful infected fixed point. Then $\text{SE}(\theta )\subseteq \text{Reg}
(S^{\ast },T^{\ast })$.
\end{proposition}

\begin{proof}
This follows from noting that $j\in {\mathop{\textup {SE}}}(\theta )$, $
S_{j}^{\ast }>0$ and $T_{j}^{\ast }=0$ implies $\dot{S_{j}}>0$. \qed
\end{proof}

\begin{proposition}
\label{prop:seStarves}\footnote{We take the opportunity to amend Proposition 7
  of Section 8 in \cite{cyclicPathogen}. The condition $R_{0}>1$ in the
  statement of that proposition is not required, given that we take SE$(\theta
  )\neq \emptyset $ as a standing assumption for the entire Section 8.}
Suppose $\text{SE}(\theta )\neq \emptyset $. Then $ \starves$ is a strict
partial order.
\end{proposition}

\begin{proof}
We must show that $\succ $ is anti-reflexive, asymmetric and transitive. The
first follows immediately from the fact that $M_{jj}=1$.

To see that $\succ $ is asymmetric, suppose we have $j\succ k$ and $k\succ j$
. This implies $(j,k]\cap $SE$(\theta )=\emptyset $ and $(k,j]\cap $SE$
(\theta )=\emptyset $, contradicting SE$(\theta )\neq \emptyset $.

To see the third we suppose that $j\succ k$ and $k\succ \ell $. We consider
two cases, $k\in \lbrack j,\ell ]$ and $\ell \in \lbrack j,k]$. In the first
case, we have $M_{jk}M_{k\ell }=M_{j\ell }$. We then have $b_{j}M_{jk}<b_{k}$
, $b_{k}M_{k\ell }<b_{\ell }$ giving $b_{j}M_{j\ell }=b_{j}M_{jk}M_{k\ell
}<b_{k}M_{k\ell }<b_{\ell }$ as required. The second case would imply $
(j,k]\cup (k,\ell ]=[0,n)$ as well as $(j,k]\cap $SE$(\theta )=\emptyset $
and $(k,\ell ]\cap $SE$(\theta )=\emptyset $, contradicting SE$(\theta )\neq
\emptyset $. \qed
\end{proof}

We define $\str{(\theta )}$, $\unstr{(\theta )}$, saturated and moderated as
before.  Proposition \ref{prop:starves} holds in the case
$\se(\theta)\ne\emptyset$.  However, there is a small change in the proof.  In
the case where $\se(\theta)=\emptyset$, we appeal to observation 10).  In the
case where $\se(\theta)\ne\emptyset$, we need to note that $j\starves k$
implies that $k\notin \se(\theta)$ and we are thus able to appeal to
observation 10$'$.  The proof then proceeds as before.

The equivalence of moderation and saturation (Proposition \ref
{prop:saturatedEqualsModerated}) and their necessity for stability (Theorem
\ref{thm:unsaturatedUnstable}) can be proved as before, the result of
Proposition \ref{prop: SelfEstablishmentImpliesRegulation} playing an
important role.

In order to prove sufficiency in the presence of self-establishing stages
(Theorem \ref{thm:uniqueStableFixedPoint} and its consequences), we rely on
Lemma 1 of Section 8 in \cite{cyclicPathogen} and the argument provided there
after the Proof of the Lemma.

\section{Discussion}

\label{sec:discussion}

In this paper, we have generalized the T-cell activation and
proliferation model of \cite{cyclicPathogen} in order to make that
model applicable to more realistic antigen dose - T-cell proliferation
response curves. While in most regimes, we would expect T-cell
proliferation to rise in response to increased antigen, this rate is
not driven by the encounter of T-cells with infected target cells, but
rather by the presentation of antigen to T-cells by dendritic cells
(\cite{janeway}). Thus, a model of the mechanisms underlying T-cell
proliferation must include additional cell populations and quite
complicated cellular processes carried out by those cells
(\cite{Zhang}).

Further, there is a widespread phenomenon that cannot be explained as
the fixed point of a system like (\ref{eq:workingSaturationModel}),
namely the existence of long-lived T-cell responses to multiple
epitopes, either of a single pathogen or in our case to a single
pathogen stage. Multiple T-cell responses are often modeled as
competing for antigen in a predator-prey dynamic. The antigenicity of
the T-cell's epitope functions as the T-cell's fitness and this leads
to a winner-take-all dynamic where the response to the most antigenic
epitope survives and the others become extinct (\cite {NowakMay},
\cite{killerCellDynamics}). It seems likely that memory T-cells play an
important role in the survival of multiple responses. However, the
interactions between effector and memory populations are still not
well understood (\cite{generationOfEffector}, \cite{roleOfModels}).
These populations exhibit differing longevity.  Stated in terms of our
model, they do not share a common value for $b$.  In addition, high
levels of antigen can lead to CTL exhaustion (\cite{Mueller}), a
phenomenon that argues against a monotone increasing $\phi$.

The upshot of this is that if we wish to refine the cyclic pathogen model to
present an increasingly detailed picture of the cell populations and their
mechanisms, we will need to include multiple immune cell populations at each
stage. The dynamics of such a system could be quite complicated. However,
there is a variable that summarizes the collective immune pressure against a
given stage, namely their net kill rate of the effector populations. Thus,
if $T_{j1},\dots,T_{jk}$ are the effector populations against stage $S_j$,
we can write $\tau_j =\tau_j(T_{j1},\dots,T_{jk})$ so that we now have 
\begin{equation*}
\frac{d S_j}{dt} = r_{j-1} f_{j-1} S_{j-1} - (a_j + f_j + \tau_j) S_j. 
\end{equation*}
We cannot expect that the dynamics of such an expanded system can be
mapped to the dynamics of (\ref{eq:workingSaturationModel}) because we
cannot necessarily expect $T_{j1},\dots,T_{jk}$ (and any non-effector
populations) to vary in a way which makes $\frac{d \tau_j}{dt}$ a
function of $S_j$ and $ \tau_j$. However, once the dynamics of these
populations are understood, in the neighborhood of a fixed point,
understanding the marginal response of $\tau_j $ to $S_j$ may allow us
to use (\ref{eq:workingSaturationModel}) to summarize these dynamics
in a way which will allow us to establish the existence of a stable
fixed point.

\section*{Acknowledgements}

We wish to thank Dr.\ Jared Hawkins and Prof.\ David Thorley-Lawson
for many conversations about the underlying biology which inspired
this work. Dr.\ Hawkins was especially helpful in bringing many apropos
references to our attention.

\bibliographystyle{plain}


\end{document}